\documentclass[12pt]{amsart}
\usepackage{amsmath, amssymb, amsthm, mathrsfs, mathtools, amscd}
\usepackage[latin1]{inputenc}
\usepackage{hyperref}
\usepackage[all]{xy}

\newtheorem{theorem}{Theorem}[section]
\newtheorem{lemma}[theorem]{Lemma}

\newtheorem{proposition}[theorem]{Proposition}
\newtheorem{definition}[theorem]{Definition}
\newtheorem{remark}[theorem]{Remark}

\newcommand\cN{\mathcal{N}}

\newcommand\PN{\mathcal{P}(\cN)}

\newcommand\hA{{\hat A}}
\newcommand\hB{\hat B}
\newcommand\hE{\hat E}
\newcommand\EA{E^{\hA}}
\newcommand\hEA{\hE^{\hA}}
\newcommand\bbR{\mathbb{R}}
\newcommand\bbN{\mathbb{N}}
\newcommand\ra{\rightarrow}
\newcommand\mt{\mapsto}
\newcommand\lra{\longrightarrow}
\newcommand\lmt{\longmapsto}
\newcommand\bmeet{\bigwedge}
\newcommand\meet{\wedge}
\newcommand\bjoin{\bigvee}
\newcommand\join{\vee}
\newcommand\oA{o^{\hA}}
\newcommand\eR{\overline{\bbR}}
\newcommand\hP{{\hat P}}
\newcommand\hQ{\hat Q}
\newcommand\cH{\mathcal{H}}
\newcommand\Eo{E^o}

\renewcommand\sp{\operatorname{sp}}
\newcommand\sa{\operatorname{sa}}
\newcommand\mc[1]{\mathcal{#1}}

\newcommand\doto[2]{\delta^o({#1})_{#2}}
\newcommand\pde[1]{\ps{\delta(#1)}}
\newcommand\cA{\mathcal{A}}
\newcommand\eq[1]{(\ref{#1})}
\newcommand\PzN{\mathcal{P}_0(\cN)}
\newcommand\ol[1]{\overline{#1}}
\newcommand\eA{e^{\hat A}}

\newcommand\VlA{V_{\hA}}

\newcommand\Cl{\mc{C}l}
\newcommand\ps[1]{\underline{#1}}
\newcommand\Sig{\ps{\Sigma}}
\newcommand\on[1]{\operatorname{#1}}
\newcommand\Subcl[1]{\on{Sub}_{\on{cl}}#1} 	% clopen subobjects
\newcommand\ld{\lambda}
\newcommand\bbC{\mathbb{C}}
\newcommand\VN{\mc V(\cN)}

\newcommand\PV{\mc P(V)}
\newcommand\im{\on{im}}
\newcommand\Set{\mathbf{Set}}
\newcommand\SetC[1]{\Set^{{#1}^{\rm op}}}
\newcommand\SetVNop{\SetC{\mathcal{V(N)}}}
\newcommand\bdas[1]{\breve\delta(#1)}
\newcommand\Rlr{\ps{\bbR^{\leftrightarrow}}}
\newcommand\De{\Delta}
\newcommand\de{\delta}

\newcommand\Ain{A\varepsilon}
\newcommand\EbrA{E^{\breve{A}}}
\newcommand\obrA{o^{\breve{A}}}
\newcommand\tCA{\tilde C^A}
\newcommand\tqA{\tilde q^A}
\newcommand\cB{\mc B}
\newcommand\tr{\on{tr}}
\newcommand\trho{\tilde\rho}

\begin{document}
\title[Self-adjoint Operators as Functions II]{Self-adjoint Operators as Functions II:\\Quantum Probability}

\author{Andreas D\"oring}
\address{Andreas D\"oring\newline
\indent Clarendon Laboratory\newline
\indent Department of Physics\newline%
\indent University of Oxford\newline
\indent Parks Road\newline
\indent Oxford OX1 3PU, UK}
\email{doering@atm.ox.ac.uk}
\author{Barry Dewitt}
\address{Barry Dewitt\newline
\indent Department of Engineering and Public Policy\newline
\indent Carnegie Mellon University\newline
\indent 5000 Forbes Avenue\newline
\indent Pittsburgh, PA, 15213\newline
\indent USA}
\email{barrydewitt@cmu.edu}
\date{December 5, 2013}

\begin{abstract}
In ``Self-adjoint Operators as Functions I: Lattices, Galois Connections, and the Spectral Order'' \cite{DoeDew12a}, it was shown that self-adjoint operators affiliated with a von Neumann algebra $\cN$ can equivalently be described as certain real-valued functions on the projection lattice $\PN$ of the algebra, which we call $q$-observable functions. Here, we show that $q$-observable functions can be interpreted as generalised quantile functions for quantum observables interpreted as random variables. More generally, when $L$ is a complete meet-semilattice, we show that $L$-valued cumulative distribution functions (CDFs) and the corresponding $L$-quantile functions form a Galois connection. An ordinary CDF can be written as an $L$-CDF composed with a state. For classical probability, one picks $L=\cB(\Omega)$, the complete Boolean algebra of measurable subsets modulo null sets of a measurable space $\Omega$. For quantum probability, one uses $L=\PN$, the projection lattice of a nonabelian von Neumann algebra $\cN$. Moreover, using some constructions from the topos approach to quantum theory, we show that there is a joint sample space for \emph{all} quantum observables, despite no-go results such as the Kochen-Specker theorem. Specifically, the spectral presheaf $\Sig$ of a von Neumann algebra $\cN$, which is not a mere set, but a presheaf (i.e., a `varying set'), plays the role of the sample space. The relevant meet-semilattice in this case is $L=\Subcl\Sig$, the complete bi-Heyting algebra of clopen subobjects of $\Sig$. We show that using the spectral presheaf $\Sig$ and associated structures, quantum probability can be formulated in a way that is \emph{structurally} very similar to classical probability.
\end{abstract}

\maketitle


%--------------------------------------------------------
\section{Introduction}
Let $\cN$ be a von Neumann algebra, let $\PN$ be its projection lattice, and let $\eR$ denote the extended reals. A $q$-observable function is a function
\begin{equation}
			o:\PN\lra\eR
\end{equation}
that preserves joins and is characterised by two simple properties: (a) $o(\hP)>-\infty$ for all non-zero projections $\hP\in\PN$, and (b) there is a family $(\hP_i)_{i\in I}\subseteq\PN$ with $\bjoin_{i\in I}\hP_i=\hat 1$ such that $o(\hP_i)\lneq\infty$ for all $i\in I$. Nothing in this description alludes to linear self-adjoint operators, but each $q$-observable function determines a unique self-adjoint operator affiliated with $\cN$ and vice versa, as was shown in ``Self-adjoint Operators as Functions I: Lattices, Galois Connections, and the Spectral Order'' \cite{DoeDew12a} (in the following called Part I). In section \ref{Sec_ResultsPartI}, we recall some results of Part I that will play a role in this article.

In section \ref{Sec_qObsFctsAndProbab}, we provide an interpretation of $q$-observable functions in terms of probability theory. They turn out to be quantile functions of projection-valued measures arising from self-adjoint operators. Some general theory of quantile functions with values in a complete meet-semilattice $L$ is developed and applied to quantum probability theory, extending the usual description of quantum probability based on structures in nonabelian von Neumann algebras.

An even closer relationship and structural similarity between classical and quantum probability is established in section \ref{Sec_QProbabAndSigma} by showing that there is a joint sample space for all quantum observables in the form of the spectral presheaf of a von Neumann algebra, notwithstanding no-go results such as the Kochen-Specker theorem. The latter only applies to sample spaces which are sets, while we consider `generalised sets' in the form of presheaves. We will show how random variables, probability measures, cumulative distribution functions and quantile functions can be formulated in a presheaf-based perspective, with the spectral presheaf taking the role of the sample space. Section \ref{Sec_Summary} concludes and gives a short outlook on future work.

%--------------------------------------------------------
\section{Some results of Part I}			\label{Sec_ResultsPartI}
In order to make the present article largely self-contained, we briefly introduce some definitions and summarise some results of Part I that are relevant for the following sections. For more details and proofs, please see Part I \cite{DoeDew12a}.

\textbf{Self-adjoint operators and $q$-observable functions.} Let $\cN$ be a von Neumann algebra, and let $\hA$ be a self-adjoint operator affiliated with $\cN$. If $\hA$ is bounded, then it lies in $\cN_{\sa}$, the self-adjoint operators in $\cN$. The set of self-adjoint operators affiliated with $\cN$ is denoted $SA(\cN)$. Clearly, $\cN_{\sa}\subset SA(\cN)$.

A \emph{$q$-observable function} on $\PN$ is a function
\begin{equation}
			o:\PN \lra \eR
\end{equation}
from the projections in $\cN$ to the extended reals $\eR=\bbR\cup\{-\infty,\infty\}$ that preserves joins\footnote{That is, $o(\bjoin_{i\in I}\hP_i)=\sup_{i\in I}o(\hP_i)$ for all families $(\hP_i)_{i\in I}\subseteq\PN$.} such that
\begin{itemize}
	\item [(a)] $o(\hP)>-\infty$ for all non-zero projections $\hP\in\PN$,
	\item [(b)] there is a family $(\hP_i)_{i\in I}\subseteq\PN$ with $\bjoin_{i\in I}\hP_i=\hat 1$ such that $o(\hP_i)\lneq\infty$ for all $i\in I$.
\end{itemize}
The set of all $q$-observable functions on $\PN$ is denoted $QO(\PN,\eR)$. In \cite{DoeDew12a}, it was shown that there is a bijection
\begin{equation}
			SA(\cN) \lra QO(\PN,\eR)
\end{equation}
between self-adjoint operators affiliated with $\cN$ and $q$-observable functions. This works in two steps: a self-adjoint operator $\hA\in SA(\cN)$ has a right-continuous spectral family $\EA:\bbR\ra\PN$ that can canonically be extended to the extended reals $\eR$ by setting $\hEA_{-\infty}:=\hat 0$ and $\hEA_{\infty}=\hat 1$. Then
\begin{equation}
			\EA:\eR \lra \PN
\end{equation}
is a monotone function between complete meet-semilattices which preserves all meets, so by the adjoint functor theorem for posets (see e.g. Example 9.33 in \cite{Awo10} or Thm. 2.5 in \cite{DoeDew12a}), $\EA$ has a left adjoint
\begin{align}
			\oA:\PN &\lra \eR\\			\nonumber
			\hP &\lmt \inf\{r\in\eR \mid \hP\leq\hEA_r\}.
\end{align}
It is straightforward to show that $\oA$ is a $q$-observable function.

Conversely, a $q$-observable function $o$ has a right adjoint $\Eo:\eR\ra\PN$, which can be shown to be a right-continuous extended spectral family $\EA:\eR\ra\PN$. For details, see Prop. 3.7 and Thm. 3.8 in \cite{DoeDew12a}.

If $\oA$ is the $q$-observable function corresponding to some self-adjoint operator $\hA$ and $\PzN$ denotes the non-zero projections in $\cN$, then
\begin{equation}
			\oA(\PzN)=\sp\hA,
\end{equation}
as shown in Lemma 3.12 of \cite{DoeDew12a}. Clearly, the operator $\hA$ is bounded if and only if $\oA(\PzN)$ is compact. Thus, the bounded self-adjoint operators in $\cN_{\sa}$ correspond bijectively to the $q$-observable functions with compact image on non-zero projections (see Prop. 3.15 in \cite{DoeDew12a}).

\textbf{Spectral order.} The set $SA(\cN)$ can be equipped with the \emph{spectral order} \cite{Ols71,deG04}: let $\hA,\hB\in SA(\cN)$, then
\begin{equation}
			\hA\leq_s\hB \quad:\Longleftrightarrow\quad (\forall r\in\eR:\hEA_r\geq\hat E^{\hB}_r).
\end{equation}
One can show that $(SA(\cN),\leq_s)$ is a conditionally complete lattice. Moreover, the set $QO(\PN,\eR)$ of $q$-observable functions can be equipped with the pointwise order, which also makes it a conditionally complete lattice.

As shown in Prop. 3.9 in Part I, there is an order-isomorphism
\begin{align}
			\phi:(SA(\cN),\leq_s) &\lra (QO(\PN,\eR),\leq)\\			\nonumber
			\hA &\lmt \oA,
\end{align}
that is, the spectral order on self-adjoint operators corresponds to the pointwise order on $q$-observable functions.

\textbf{Rescalings.} If $f:\eR\ra\eR$ is a join-preserving function, then, for all $\hA\in SA(\cN)$, it holds that
\begin{equation}
			o^{f(\hA)}=f(\oA),
\end{equation}
which provides a limited form of functional calculus for $q$-observable functions (see Thm. 4.7 in \cite{DoeDew12a}). Physically, such a function $f$ can be interpreted as an (order-preserving) rescaling of the values of outcomes. In simple cases, this corresponds to a change of units, e.g. from $^\circ C$ to $^\circ F$.

%--------------------------------------------------------
\section{$q$-observable functions and probability theory}			\label{Sec_qObsFctsAndProbab}
Subsection \ref{Subsec_ClassProbab} provides a very brief recap of some mathematical structures in classical probability theory, including cumulative distribution functions (CDFs, given by \emph{c\`adl\`ags}) and quantile functions. As far as we are aware, the observation that a CDF and the corresponding quantile function form a Galois connection is new.

In subsection \ref{Subsec_LValuedMeasures}, we consider more general measures with values in a complete meet-semilattice $L$ and define the new notions of \emph{$L$-CDFs} and \emph{$L$-quantile functions}, which again are related to each other by an adjunction (i.e., a Galois connection). The generalisation to $L$-valued measures and associated notions is useful in quantum theory, where one deals with projection-valued measures, and also serves as a preparation for section \ref{Sec_QProbabAndSigma}, where we will consider measures taking values in a certain complete bi-Heyting algebra of clopen subobjects of a presheaf.

The step to quantum theory and noncommutative probability is taken in subsection \ref{Subsec_QProbab}. We present some standard material on the use of von Neumann algebras in classical and quantum probability and show that $q$-observable functions are $\PN$-quantile functions for $\PN$-valued measures, adding an aspect to the established picture. Spectral families play the role of $\PN$-CDFs. We emphasise the role of the Gelfand spectrum as a sample space in the commutative case, because the spectral presheaf that will be introduced in section \ref{Sec_QProbabAndSigma} is a generalisation of the Gelfand spectrum to nonabelian von Neumann algebras and will serve as a joint sample space for all (noncommuting) quantum random variables. 

\subsection{Basic structures in classical probability}			\label{Subsec_ClassProbab}
We first describe some basic structures in classical probability theory. There are many good references, though quantile functions are not treated in some standard texts. For these, see e.g. \cite{Was04,Gil00}.

Let $(\Omega,\cB(\Omega),\mu)$ be a measure space, where $\Omega$ is a non-empty set, $\cB(\Omega)$ is a $\sigma$-algebra of $\mu$-measurable subsets of $\Omega$, and $\mu$ is a probability measure. The elements of $\cB(\Omega)$ are called \emph{events}, and $\Omega$ is called the \emph{sample space} of the system. The probability measure $\mu$ is a map $\mu:\cB(\Omega)\ra [0,1]$ such that $\mu(\Omega)=1$, and for all countable families $(S_i)_{i\in\bbN}\subset B(\Omega)$ of pairwise disjoint events, it holds that
\begin{equation}
			\mu(\bigcup_{i\in\bbN}S_i)=\sum_{i\in\bbN}\mu(S_i).
\end{equation}

A \emph{random variable} is a measurable function $A:\Omega\ra\im A\subseteq\bbR$. Conceptually, the inverse image function $A^{-1}$ is more important. Since $A$ is measurable, we have
\begin{equation}
			A^{-1}(\De)\in\cB(\Omega)
\end{equation}
for every Borel subset $\De$ of $\bbR$. In this way, $A^{-1}$ maps measurable subsets of numerical outcomes to measurable subsets of the sample space.

The \emph{cumulative distribution function (CDF) $C^A$} of a random variable $A$ with respect to a given probability measure $\mu$ is defined as
\begin{align}			\label{Def_CA}
			C^A:\bbR &\lra [0,1]\\			\nonumber
			r &\lmt \mu(A^{-1}(-\infty,r]).
\end{align}
We can canonically extend $C^A$ to a function from the extended reals $\eR$ to $[0,1]$ by setting $C^A(-\infty):=0$ and $C^A(\infty):=1$. Clearly, an extended CDF is uniquely determined by its restriction to $\bbR$, so we will not distinguish between CDFs (defined on $\bbR$) and extended CDFs (defined on $\eR$) notationally. In the following, CDF will always mean extended CDF unless mentioned otherwise.

The measure $\mu$ assigns probability weights to all events $S\in\cB(\Omega)$. Consider the event representing the random variable $A$ having values not exceeding $r$, which is $A^{-1}((-\infty,r])$. Then $C^A(r)$ is the probability weight of this event. So, roughly speaking, the value $C^A(r)$ of the CDF at $r$ expresses the probability that $A$ has a value not exceeding $r$.

It is obvious that a CDF $C^A:\eR\ra [0,1]$ is order-preserving and right-continuous. Moreover, $\lim_{r\ra -\infty}C^A(r)=0$ and $\lim_{r\ra\infty}C^A(r)=1$. This means that CDFs are examples of \emph{c\`adl\`ag functions} (also simply called \emph{c\`adl\`ags}), which stands for \emph{continue \`a droite, limite \`a gauche}.\footnote{The set of c\`adl\`ags is typically denoted as $D$ and can be equipped with the Skorohod topology, which we will not introduce here. For details, see e.g. \cite{Bil99}.}

We observe that each CDF $C^A:\eR\ra [0,1]$ on the extended reals is a map between complete meet-semilattices that preserves all meets and hence has a left adjoint
\begin{align}
			q^A:[0,1] &\lra \eR\\			\nonumber
			p &\lmt \inf\{r\in\eR \mid C^A(r)\geq p\}.
\end{align}
This function is well-known: it is the \emph{quantile function of $A$} with respect to $\mu$. 

\begin{remark}
The existence of the Galois connection between a CDF $C^A$ and the corresponding quantile function $q^A$ seems to have been overlooked in the literature, though it is implicit in the usual treatments. The adjunction is available because $C^A$ can canonically be extended to $\eR$, which makes $C^A:\eR\ra [0,1]$ into a meet-preserving map between complete meet-semilattices and allows us to employ the adjoint functor theorem for posets (see e.g. \cite{Awo10,DoeDew12a}). The left adjoint $q^A:[0,1]\ra\eR$ preserves all joins, i.e., suprema. (Recall that each complete meet-semilattice is also a complete join-semilattice.) Usually, $q^A$ is defined as a function from $(0,1)$ to $\bbR$ that is left-continuous. Leaving out the probabilities $0$ and $1$ is artificial.
\end{remark}

The property of left-continuity of a quantile function becomes a simple consequence of the existence of the Galois connection: being a left adjoint, the quantile function $q^A$ preserves joins, so in particular
\begin{equation}
				\forall p\in [0,1]: \sup_{s<p} q^A(s) = q^A(\sup_{s<p} s) = q^A(p).
\end{equation}

Note that since $C^A(-\infty)=0$, we have $q^A(0)=-\infty$. For any $p>0$, it holds that $q^A(p)\gneq-\infty$, and $q^A(1)=\inf\{r\in\eR \mid C^A(r)\geq 1\}=\sup(\im A)$.

The function $q^A$ assigns to each probability $p\in (0,1]$ the smallest value $r$ (which is easily seen to lie in the image of the random variable $A$) such that the probability of $A$ having a value not exceeding $r$ is (at least) $p$. Intuitively speaking, the quantile function $q^A$ is the `inverse' of the CDF $C^A$.

\subsection{$L$-valued measures, $L$-CDFs and $L$-quantile functions}			\label{Subsec_LValuedMeasures}
Let $(\Omega,\cB(\Omega))$ be a measurable space, and let $A:\Omega\ra\eR$ be a random variable as before (where $\im A\subseteq\bbR\subset\eR$). We want to define a function $\tCA$ similar to the cumulative distribution function $C^A:\eR\ra [0,1]$, but without any reference to a probability measure $\mu$. To this end, we think of the inverse image $A^{-1}$ of the random variable as a `$\cB(\Omega)$-valued measure'. 

We will see that it is advantageous to assume that $\cB(\Omega)$ is a complete Boolean algebra, not just a $\sigma$-complete one. A standard result shows that if $\cB_{\sigma}(\Omega)$ is a $\sigma$-complete Boolean algebra, and if one is given a probability measure $\mu:\cB(\Omega)\ra [0,1]$, then
\begin{equation}
			\cB(\Omega) := \cB_{\sigma}(\Omega)/\mc N(\Omega)
\end{equation}
is a complete Boolean algebra, where $\mc N(\Omega)$ denotes the subsets of measure $0$ with respect to $\mu$. 

In fact, for quantum theory we will consider more general lattices $L$ than just complete Boolean algebras of the form $\cB(\Omega)$. For the generalised notions of CDFs and quantile functions, we merely need that $L$ is a complete meet-semilattice (and hence also is a complete join-semilattice). $L$ need not be distributive.

We equip the extended reals $\eR$ with the topology generated by the set ${\eR}$ itself and by the intervals $(r,s)$, where $r,s\in\eR$, $r\leq s$ and $(r,r)=\emptyset$ by convention. The relative topology of $\bbR\subset\eR$ is the standard topology on $\bbR$. By $\cB(\eR)$, we denote the $\sigma$-complete Boolean algebra of Borel subsets of $\eR$.

\begin{definition}
Let $\cB(\Omega)$ be a complete Boolean algebra of (equivalence classes of) subsets of a set $\Omega$, and let $A^{-1}:\cB(\eR)\ra\cB(\Omega)$ be the inverse image of a random variable $A:\Omega\ra\im A\subset\eR$ such that $A^{-1}$ preserves all existing meets. The \emph{$\cB(\Omega)$-cumulative distribution function ($\cB(\Omega)$-CDF) $\tCA$ of $A^{-1}$} is the function
\begin{align}			\label{Def_tCAForB(Om)}
			\tCA: \eR &\lra \cB(\Omega)\\			\nonumber
			r &\lmt A^{-1}([-\infty,r]).
\end{align}
\end{definition}

More generally, we can replace $\cB(\Omega)$ in the definition above by a complete meet-semilattice $L$, and assume that $A^{-1}:\cB(\eR)\ra L$ preserves all existing meets. This implies in particular that
\begin{equation}
	A^{-1}(\eR)=A^{-1}(\bmeet\emptyset)=\bmeet\emptyset=\top_L,
\end{equation}
the top element of $L$. Moreover, 
\begin{equation}
			A^{-1}(\emptyset)=A^{-1}(\bmeet_{S\in\cB(\eR)} S)=\bmeet_{S\in\cB(\eR)} A^{-1}(S)\supseteq\bmeet_{T\in L} T=\bot_L,
\end{equation}
the bottom element of $L$. We assume that $A^{-1}$ is such that $A^{-1}(\emptyset)=\bot_L$. This generalises the classical case: if $L=\cB(\Omega)$, then $A^{-1}(\emptyset)=\emptyset=\bot_{\cB(\Omega)}$.

$A^{-1}$ is interpreted as an \emph{$L$-valued measure}, and the \emph{$L$-CDF of $A^{-1}$} is the function
\begin{align}			\label{Def_L-CDF}
			\tCA: \eR &\lra L\\			\nonumber
			r &\lmt A^{-1}([-\infty,r]).
\end{align}
Note that here, $A^{-1}$ is a symbolic notation, since there may be no function $A$ such that $A^{-1}$ is its inverse image function. In the definition of $\tCA$, we can include $-\infty$ in the interval $[-\infty,r]$, since $A^{-1}$ is defined on $\cB(\eR)$, the Borel subsets of $\eR$ (and not just on $\cB(\bbR)$, the Borel subsets of $\bbR$).

Let $(r_i)_{i\in I}\subseteq\eR$ be an arbitrary net of extended real numbers. It is easy to see that
\begin{equation}			\label{Eq_InfsToIntersects}
			[-\infty,\inf_{i\in I} r_i] = \bigcap_{i\in I} [-\infty,r_i] \in \cB(\eR).
\end{equation}
In particular, the meet (intersection) on the right-hand side exists in the $\sigma$-complete Boolean algebra $\cB(\eR)$ even if $I$ is an uncountable index set or net. The assumption that $A^{-1}:\cB(\eR)\ra\cB(\Omega)$ preserves all existing meets means that
\begin{equation}			\label{Eq_A-1(Meet)IsMeet(A-1)}
			A^{-1}(\bigcap_{i\in I} S_i) = \bmeet_{i\in I} A^{-1}(S_i)
\end{equation}
for all families $(S_i)_{i\in I}\subseteq\cB(\eR)$ such that $\bigcap_{i\in I} S_i\in\cB(\eR)$. In particular,
\begin{equation}			\label{Eq_AInversePresMeets}
			A^{-1}([-\infty,\inf_{i\in I} r_i]) \stackrel{\eq{Eq_InfsToIntersects}}{=} A^{-1}(\bigcap_{i\in I}[-\infty,r_i]) \stackrel{\eq{Eq_A-1(Meet)IsMeet(A-1)}}{=} \bmeet_{i\in I} A^{-1}([-\infty,r_i]).
\end{equation}

\begin{lemma}
Let $L$ be a complete meet-semilattice, and let $A^{-1}:\cB(\eR)\ra L$ be an $L$-valued measure that preserves all existing meets. Then the $L$-CDF $\tCA:\eR\ra L$ of $A^{-1}$ preserves all meets.
\end{lemma}

\begin{proof}
We have, for all families $(r_i)_{i\in I}\subseteq\eR$,
\begin{align}
			\tCA(\inf_{i\in I} r_i) &= A^{-1}([-\infty,\inf_{i\in I} r_i])\\
			&\stackrel{\eq{Eq_AInversePresMeets}}{=} \bmeet_{i\in I} A^{-1}([-\infty,r_i])\\
			&= \bmeet_{i\in I} \tCA(r_i).
\end{align}
\end{proof}

Here we make use of the fact that $\eR$ is a complete meet-semilattice, while $\cB(\eR)$ is just $\sigma$-complete. Since $\tCA$ preserves meets, it has a left adjoint:
\begin{definition}
Let $\tCA:\eR\ra L$ be the $L$-CDF of an $L$-valued measure $A^{-1}:\cB(\eR)\ra L$. The \emph{$L$-quantile function $\tqA$ of $A^{-1}$} is the left adjoint of $\tCA$,
\begin{align}
			\tqA: L &\lra \eR\\			\nonumber
			T &\lmt \inf\{r\in\eR \mid T\leq\tCA(r)\}.
\end{align}
\end{definition}
The map $\tqA$ between the complete join-semilattices $L$ and $\eR$ preserves all joins by the adjoint functor theorem for posets.

\noindent\textbf{Relation to standard CDFs and quantile functions.} Let $A:\Omega\ra\eR$ be a random variable, and let $A^{-1}:\cB(\eR)\ra\cB(\Omega)$ be the corresponding $\cB(\Omega)$-valued measure. If $\tCA$ is the $\cB(\Omega)$-CDF of $A^{-1}$, and $\mu:\cB(\Omega)\ra [0,1]$ is a probability measure, then the usual CDF of $A$ with respect to $\mu$ is
\begin{equation}
			C^A=\mu\circ\tCA: \eR \lra \cB(\Omega) \lra [0,1],
\end{equation}
as can be seen from the definition of $\tCA$, \eq{Def_tCAForB(Om)}, and the definition of $C^A$, \eq{Def_CA}. (Note that now $A^{-1}$ is defined on $\cB(\eR)$, not just on $\cB(\bbR)$ as in \eq{Def_CA}, so we can write $A^{-1}([-\infty,r])$ instead of $A^{-1}((-\infty,r])$.)

This generalises directly to the situation of $L$-valued measures $A^{-1}:\cB(\eR)\ra L$ and probability measures $\mu: L \ra [0,1]$; for these we have the decomposition
\begin{equation}			\label{Eq_CDFDecomposed}
			C^A=\mu\circ\tCA: \eR \lra L \lra [0,1]
\end{equation}
of the usual CDF $C^A$ into the $L$-CDF $\tCA$, followed by the probability measure $\mu$. 

We had assumed that $A^{-1}$ is meet-preserving, which implies that $\tCA:\eR\ra L$ is meet-preserving as well. Hence, a CDF $C^A$ is a meet-preserving map if and only if $\mu|_{\im\tCA}:L\ra [0,1]$ preserves all meets. This can be seen as a regularity condition on the measure $\mu$.

If the probability measure $\mu$ preserves meets, then it has a left adjoint
\begin{align}
			\kappa: [0,1] &\lra L\\			\nonumber
			s &\lmt \bmeet\{T\in L \mid s\leq\mu(T)\}.
\end{align}
Moreover, in this case $C^A=\mu\circ\tCA$ preserves meets, so it has a left adjoint $q^A:[0,1]\ra\eR$, the quantile function of the random variable $A$ with respect to the measure $\mu$. Clearly, $q^A$ can be decomposed into two left adjoints as
\begin{equation}			\label{Eq_QuantFctDecomposed}
			q^A=\tqA\circ\kappa.
\end{equation}

\subsection{Quantum probability and projection-valued measures}			\label{Subsec_QProbab}
The quantum analogue of a random variable $A$ is usually taken to be a bounded self-adjoint operator $\hA:\cH\ra\cH$ on a Hilbert space $\cH$. We briefly recall, without proofs, how this connects with the classical case, where a random variable $A:\Omega\ra\bbR$ is a measurable function from a sample space $\Omega$ to the reals. Good references are \cite{KR83,Rud73}. The article \cite{RedSum07} gives a nice general introduction to quantum probability.

\noindent\textbf{Classical probability and abelian von Neumann algebras.} Let $\Omega$ be a measure space that is either discrete (that is, a finite or countable set, equipped with the counting measure), continuous (that is, isomorphic to $[0,1]$, equipped with the Lebesgue measure), or of mixed type, with a discrete and a continuous part. Let $L^\infty(\Omega)$ be the Banach space of equivalence classes (modulo measure-zero subsets) of complex, bounded, measurable functions on $\Omega$, with the essential supremum as norm. This is an abelian Banach algebra under pointwise multiplication, and in fact is an abelian von Neumann algebra. Moreover, every abelian von Neumann algebra is isomorphic to an algebra of the form $L^\infty(\Omega)$ for $\Omega$ a measure space of one of the types listed above. 

Let $\Sigma$ denote the Gelfand spectrum of $L^\infty(\Omega)$. By Gelfand duality, there is an isometric $*$-isomorphism
\begin{equation}
			L^\infty(\Omega)\simeq C(\Sigma).
\end{equation}
Note that the Gelfand topology on $\Sigma$ is such that the equivalence classes of measurable functions on $\Omega$ correspond to continuous functions on $\Sigma$. The spectrum $\Sigma$ is an extremely disconnected compact Hausdorff space in the Gelfand topology, that is, the closure of each open set is open, and the interior of each closed set is closed. The clopen, i.e., closed and open subsets of $\Sigma$ form a complete (not just $\sigma$-complete) Boolean algebra $\Cl(\Sigma)$.

Hence, we can view the real-valued elements in $C(\Sigma)$, which are the Gelfand transforms of the real-valued elements in $L^\infty(\Omega)$, as random variables. The Gelfand spectrum $\Sigma$ is their joint sample space, with the clopen subsets $\Cl(\Sigma)$ as measurable subsets. 

Since $C(\Sigma)$ is an abelian $C^*$-algebra, it can be represented faithfully by a $C^*$-algebra $V$ of bounded operators on a Hilbert space $\cH$, and as $C(\Sigma)$ is weakly closed, i.e., a von Neumann algebra, so is $V$. Let $\ol A$ be a real-valued function in $C(\Sigma)$, and let $\hA$ be the corresponding self-adjoint operator in $V$. Then
\begin{align}
			\ol A:\Sigma &\lra \sp\hA\\			\nonumber
			\ld &\lmt \ol A(\ld):=\ld(\hA).
\end{align}
The random variable $\ol A$ has $\Sigma$ as its sample space and $\sp\hA$ as its set of outcomes. Of course, one can equivalently regard the bounded self-adjoint operator $\hA$ as the representative of the random variable.

The smallest abelian von Neumann algebra that contains a given bounded self-adjoint operator $\hA$ is $\VlA:=\{\hA,\hat 1\}''$, the double commutant of the set $\{\hA,\hat 1\}$. In this case, the random variable
\begin{align}
			\ol A:\Sigma_{\VlA} &\lra \sp\hA\\			\nonumber
			\ld &\lmt \ld(\hA).
\end{align}
is a bijection. Intuitively, the sample space $\Sigma_{\VlA}$ is the smallest one that can accommodate the random variable $\ol A$. If $V\supset\VlA$ is any bigger abelian von Neumann algebra that contains $\hA$, the random variable $\ol A:\Sigma_V\ra\sp\hA$ is a surjective, but not an injective map.

If we consider a random variable in the form of a bounded self-adjoint operator $\hA$, lying in the abelian von Neumann algebra $\VlA$, then the spectral theorem shows that the `inverse image' of this random variable is the well-known projection-valued spectral measure
\begin{equation}
			\eA:\cB(\sp\hA) \lra \mc P(\VlA),
\end{equation}
sending Borel sets of numerical outcomes (i.e., values in the spectrum of $\hA$) to projections in $\VlA$. In the notation of subsection \ref{Subsec_LValuedMeasures}, $\eA$ is a $\mc P(\VlA)$-valued measure. $L=\mc P(\VlA)$ is a complete Boolean algebra, so in particular, it is a complete meet-semilattice. If $V$ is some abelian von Neumann algebra such that $V\supset\VlA$, then there is a spectral measure $\eA:\cB(\sp\hA)\ra\PV$, given by
\begin{equation}
			\eA:\cB(\sp\hA) \ra \mc P(\VlA) \hookrightarrow \PV,
\end{equation}
where $\mc P(\VlA) \hookrightarrow \PV$ is the inclusion of the complete Boolean algebra $\mc P(\VlA)$ into the complete Boolean algebra $\PV$. This inclusion preserves all meets and joins (see e.g. Lemma 1 in \cite{DoeDew12a}).

For each abelian von Neumann algebra $V$, there is an isomorphism of complete Boolean algebras
\begin{align}			\label{Def_alphaV}
			\alpha_V:\PV &\lra \Cl(\Sigma_V)\\			\nonumber
			\hP &\lmt \{\ld\in\Sigma_V \mid \ld(\hP)=1\}.
\end{align}
Using this, the spectral measure $\eA:\cB(\sp\hA)\ra\PV$ can also be regarded as a map
\begin{equation}
			\eA:\cB(\sp\hA) \ra \PV \ra \Cl(\Sigma_V),
\end{equation}
taking values in the clopen (that is, measurable) subsets of the sample space $\Sigma_V$, the Gelfand spectrum of $V$. We extend the domain of $\eA$ canonically to $\cB(\eR)$ by
\begin{equation}			\label{Def_SpecMeasOnExtendedReals}
			\forall \De\in\cB(\eR): \eA(\De):=\eA(\De\cap\sp\hA).
\end{equation}
Written in this form, $\eA:\cB(\eR)\ra\Cl(\Sigma_V)$ is indeed the inverse image of the random variable $\ol A:\Sigma_V\ra\sp\hA\subset\eR$. In the notation of subsection \ref{Subsec_LValuedMeasures}, $\eA$ is a $\Cl(\Sigma_V)$-valued measure. $L=\Cl(\Sigma_V)$ is a complete Boolean algebra, so in particular it is a complete meet-semilattice, with meets given by
\begin{equation}
			\forall (S_i)_{i\in I}\subseteq\Cl(\Sigma_V): \bmeet_{i\in I} S_i = \on{int} \bigcap_{i\in I} S_i.
\end{equation}
This is clearly an open set, and it is also closed since it is the interior of a closed set and $\Sigma_V$ is extremely disconnected.

\noindent\textbf{Noncommutative probability and nonabelian von Neumann algebras.} So far, we have simply rewritten some aspects of classical probability (sample space, random variables and their inverse images) leading to a formulation with self-adjoint operators and abelian von Neumann algebras. The step to noncommutative quantum probability is taken by considering nonabelian von Neumann algebras. Self-adjoint operators still represent random variables, and the spectral theorem holds such that we can speak of the inverse images of random variables, in the form of spectral measures $\eA$. In fact, the spectral measure of a self-adjoint operator $\hA$ in the noncommutative case is the composite
\begin{equation}
			\eA:\cB(\eR) \ra \mc P(\VlA) \hookrightarrow \PN,
\end{equation}
where $\mc P(\VlA)\hookrightarrow\PN$ is the inclusion of the complete Boolean algebra $\mc P(\VlA)$ into the complete orthomodular lattice $\PN$. This inclusion preserves all meets and joins. $\eA$ is a $\PN$-valued measure in the noncommutative case. $L=\PN$ is a complete orthomodular lattice, so in particular it is a complete meet-semilattice.

Yet, what is not immediate in the nonabelian case is a suitable generalisation of the sample space in the form of the Gelfand spectrum $\Sigma_V$ of an abelian von Neumann algebra $V$. In section \ref{Sec_QProbabAndSigma}, we will show that such a joint sample space for \emph{all} quantum random variables exists. It is given by the so-called \emph{spectral presheaf $\Sig$} of a nonabelian von Neumann algebra $\cN$.

For now, we focus on the usual projection-valued measures of the form $\eA:\cB(\eR)\ra\PN$. Here, the Hilbert space $\cH$ on which the algebra $\cN$ is represented plays the role analogous to the sample space of the system, with the closed subspaces corresponding to the projections in $\cN$ playing the role of measurable subsets. Different from the abelian case, the subspaces (respectively projections) do not form a Boolean algebra, but a non-distributive orthomodular lattice.

The right-continuous spectral family $\EA$ of the self-adjoint operator $\hA\in\cN$ is given by
\begin{align}
			\EA:\bbR &\lra \PN\\			\nonumber
			r &\lmt \eA((-\infty,r]).
\end{align}
As in Part I \cite{DoeDew12a}, we can extend the domain of $\EA$ canonically to the extended reals $\eR$ by defining $\hEA_{-\infty}:=\hat 0$ and $\hEA_{\infty}:=\hat 1$. Then
\begin{align}
			\EA:\eR &\lra \PN\\			\nonumber
			r &\lmt \eA([-\infty,r]).
\end{align}
It is now obvious that the extended spectral family $\EA$ is the $\PN$-cumulative distribution function ($\PN$-CDF) of the projection-valued measure $\eA$, in the sense of \eq{Def_L-CDF}.

This leads to the probabilistic interpretation of the $q$-observable function
\begin{equation}
			\oA:\PN \lra \eR
\end{equation}
of $\hA$: as the left adjoint of $\EA:\eR \ra \PN$, it is the $\PN$-quantile function of the quantum random variable described by the $\PN$-valued measure $\eA$. The fact that $\PN$ is not a Boolean algebra as in the classical case is irrelevant, since the definition of the left adjoint of $\EA$ does not depend on distributivity. We just need that $L=\PN$ is a complete meet-semilattice.

The observation that quantum random variables have quantile functions, in the form of left adjoints of their spectral families, is new as far as we are aware. Just like the spectral measure $\eA$ and the spectral family $\EA$, the $q$-observable function $\oA$ characterises the quantum random variable $\hA$ completely (cf. Thm. 3 in Part I \cite{DoeDew12a}).

\noindent\textbf{States as probability measures on projections.} The projection-valued measure $\eA$ of the quantum random variable $\hA$ is not the analogue of the probability measure $\mu:\cB(\Omega)\ra [0,1]$ in the classical case. Rather, the quantum version of a probability measure is provided by a \emph{state} of the quantum system. As the generalised version of Gleason's theorem \cite{Mae90} shows, each quantum state $\rho:\cN\ra\bbC$ is equivalently given by a map
\begin{equation}
			\mu_\rho:\PN\ra [0,1]
\end{equation}
such that
\begin{itemize}
	\item [(a)] $\mu_\rho(\hat 1)=1$,
	\item [(b)] for $\hP,\hQ\in\PN$ such that $\hP\hQ=\hQ\hP=\hat 0$, it holds that $\mu_\rho(\hP+\hQ)=\mu_\rho(\hP)+\mu_\rho(\hQ)$.
\end{itemize}
Such a map $\mu_\rho$ is called a \emph{finitely additive probability measure on projections}. If $\mu_\rho$ is completely additive, then the corresponding state $\rho$ is normal. The state $\rho$ and the measure $\mu_\rho$ are related in the following way:
\begin{equation}
			\mu_\rho=\rho|_{\PN}.
\end{equation}

In the following, we will only consider normal quantum states. As is well-known, every normal state $\rho$ is of the form
\begin{equation}
			\forall\hA\in\PN: \rho(\hA)=\tr(\trho\hA)
\end{equation}
for some positive trace-class operator $\trho$ of trace $1$. Normality of $\rho$ is equivalent to $\mu_\rho=\rho|_{\PN}$ being a join-preserving map: let $(r_i)_{i\in I}\subseteq\eR$ be a family of (extended) real numbers. For each $n=1,2,...,\#I$, where $\#I$ is the cardinality of $I$, let $I_n$ denote an $n$-element subset of $I$ such that $I_n\subsetneq I_{n+1}$, and define
\begin{equation}
			\hP_n:=\bjoin_{j\in I_n}\hEA_{r_j}.
\end{equation}
Then $(\hP_n)_{n=1,2,...,\#I}$ is an increasing net of projections in $\PN$ converging strongly to $\bjoin_{i\in I}\hEA_{r_i}$. Since a normal state $\rho$ preserves suprema of increasing nets, we have
\begin{equation}			\label{Eq_murhoPreservesJoins}
			\mu_\rho(\bjoin_{i\in I}\hEA_{r_i})=\rho(\bjoin_{i\in I}\hEA_{r_i})=\sup_{i\in I}\rho(\hEA_{r_i})=\sup_{i\in I}\mu_\rho(\hEA_{r_i}).
\end{equation}

There is a pairing between quantum random variables and quantum states (in the form of probability measures on projections). The probability that a quantum random variable $\hA$ has a value in the Borel subset $\De\subseteq\sp\hA$ is given by
\begin{align}
			\on{Prob}(\hA,\De)=\mu_\rho(\eA(\De)).
\end{align}
This is the direct analogue of the classical
\begin{equation}
			\on{Prob}(A,\Gamma)=\mu(A^{-1}(\Gamma)),
\end{equation}
where $A$ is a random variable and $\Gamma$ is a Borel subset of $\bbR$ (or $\eR$). 

The usual CDF of a quantum random variable $\hA$ with respect to a state $\rho$ is given by
\begin{align}
			C^{\hA}=\mu_\rho\circ\EA:\eR &\lra [0,1]\\			\nonumber
			r &\lmt \mu_\rho(\hEA_r),
\end{align}
cf. \eq{Eq_CDFDecomposed}. We will now show that $C^{\hA}$ has a left adjoint. For this, we need that $C^{\hA}$ preserves meets. The extended spectral family $\EA:\eR\ra\PN$ preserves meets due to monotonicity and right-continuity (for details, see section 3 in Part I \cite{DoeDew12a}), which implies that $C^{\hA}$ preserves meets, so we just have to show that $\mu_\rho$ preserves meets.

Let $(r_i)_{i\in I}$ be an arbitrary family of real numbers. Then
\begin{align}
			\mu_\rho(\bmeet_{i\in I}\hEA_{r_i}) &= \mu_\rho(\hat 1-\bjoin_{i\in I}(\hat 1-\hEA_{r_i}))\\
			&= 1-\mu_\rho(\bjoin_{i\in I}(\hat 1-\hEA_{r_i}))\\
			&\stackrel{\eq{Eq_murhoPreservesJoins}}{=} 1-\sup_{i\in I}\mu_\rho(\hat 1-\hEA_{r_i})\\
			&=\inf_{i\in I}\mu_\rho(\hEA_{r_i}).
\end{align}
Hence, $\mu_\rho:\PN\ra [0,1]$ has a left adjoint
\begin{align}
			\kappa_\rho: [0,1] &\lra \PN\\			\nonumber
			s &\lmt \bmeet\{\hP\in\PN \mid s\leq\mu_\rho(\hP)\},
\end{align}
and $C^{\hA}=\mu_\rho\circ\EA$ has a left adjoint
\begin{align}
			q^{\hA}=\oA\circ\kappa_\rho: [0,1] &\lra \eR\\			\nonumber
			s &\lmt \oA(\kappa_\rho(s)).
\end{align}
This $q^{\hA}$ is the usual quantile function of the quantum random variable $\hA$ with respect to the normal state $\rho$.

We can also obtain the quantile function $q^{\hA}$ directly as the left adjoint of $C^{\hA}$, that is,
\begin{align}
			q^{\hA}: [0,1] &\lra \eR\\			\nonumber
			s &\lmt \inf\{r\in\eR \mid s\leq C^{\hA}(r)\}.
\end{align}

\begin{lemma}
The two expressions for the quantile function $q^{\hA}$ coincide, that is,
\begin{equation}
			\forall s\in[0,1]: (\oA\circ\kappa_\rho)(s)=\inf\{r\in\eR \mid s\leq C^{\hA}(r)\}.
\end{equation}
\end{lemma}

\begin{proof}
We have
\begin{equation}			\label{Eq_qADef1}
			\inf\{r\in\eR \mid s\leq C^{\hA}(r)\} = \inf\{r\in\eR \mid s\leq\mu_\rho(\hEA_r)\},
\end{equation}
which is the smallest $r\in\eR$ for which $\mu_\rho(\hEA_r)\geq s$. On the other hand,
\begin{equation}			\label{Eq_qADef2}
			\oA(\kappa_\rho(s))=\inf\{r\in\eR \mid \bmeet\{\hP\in\PN \mid s\leq\mu_\rho(\hP)\leq\hEA_r\}.
\end{equation}
Note that $\hP\leq\hEA_r$ implies $\mu_\rho(\hP)\leq\mu_\rho(\hEA_r)$. Conversely, $\mu_\rho(\hP)\leq\mu_\rho(\hEA_r)$ does not in general imply $\hP\leq\hEA_r$. But if we pick $\hP=\kappa_\rho(s)$, which is the \emph{smallest} projection for which $\mu_\rho(\hP)\geq s$ holds, then $\mu_\rho(\hEA_r)\geq\mu_\rho(\hP)\geq s$ implies $\hEA_r\geq\hP=\kappa_\rho(s)$. This shows that \eq{Eq_qADef1} and \eq{Eq_qADef2} coincide.
\end{proof}

\textbf{Summary.} We summarise the relations between classical and quantum probability in table 1. On the classical side, we assume that $\cB(\Omega)$ is a complete Boolean algebra, not just a $\sigma$-complete one.

The relevant complete meet-semilattices are $L=\cB(\Omega)$ on the classical side and $L=\PN$ on the quantum side. Note that as discussed above, on the classical side we can replace 
\begin{itemize}
	\item [(i)] the sample space $\Omega$ by the Gelfand spectrum $\Sigma_V$ of some abelian von Neumann algebra $V$,
	\item [(ii)] each random variable $A$ by the Gelfand transform $\ol A:\Sigma_V\ra\sp\hA\subset\eR$ of some (bounded) self-adjoint operator $\hA\in V$, 
	\item [(iii)] the complete Boolean algebra $\cB(\Omega)$ by the complete Boolean algebra $\Cl(\Sigma_V)$.
\end{itemize}
This emphasises the role of the Gelfand spectrum $\Sigma_V$ as a sample space in the classical case.

\begin{table}[htdp]
\begin{center}			\label{Table_CAndQProbab1}
\begin{tabular}{| c | c | c |} \hline
 & & \\
 & Classical & Quantum \\
 & & \\
\hline & & \\
Sample space & $\Omega$ & $\cH$\\
 & & \\
\hline & & \\
Random variable & $A:\Omega\ra\im A\subset\eR$ & $\hA\in \cN_{\sa}$\\
 & & \\
\hline & & \\
Inv. im. of random var. & $A^{-1}:\cB(\eR)\ra\cB(\Omega)$ & $e^{\hA}:\cB(\eR) \ra \PN$\\
 & & \\
\hline & & \\
$L$-CDF & $\tCA:\eR \ra \cB(\Omega)$ & $\EA:\eR\ra\PN$\\
 & $r \mt A^{-1}([-\infty,r])$ & $r \mt \hEA_r=\eA([-\infty,r])$\\
 & & \\
\hline & & \\
$L$-quantile function & $\tqA:\cB(\Omega)\ra\eR$ & $\tilde q^{\hA}=\oA:\PN\ra\eR$\\
 & $S \mt \inf\{r\in\eR \mid S\leq\tCA(S)\}$ & $\hP \mt \inf\{r\in\eR \mid \hP\leq \hEA_r\}$\\
 & & \\
\hline & & \\
State & $\mu:\cB(\Omega)\ra [0,1]$ & $\mu_\rho:\PN\ra [0,1]$\\
 & & \\
\hline & & 
\\CDF & $C^A:\eR\ra [0,1]$ & $C^{\hA}:\eR\ra [0,1]$\\
 & $r \mt \mu(A^{-1}([-\infty,r]))$ & $r \mt \mu_\rho(\eA([-\infty,r]))$ \\
 & & \\
\hline & & \\
Quantile function & $q^A:[0,1]\ra\eR$ & $q^{\hA}:[0,1]\ra\eR$\\
 & $s \mt \inf\{r\in\eR \mid s\leq C^A(r)\}$ & $s \mt \inf\{r\in\eR \mid s\leq C^{\hA}(r)\}$ \\
 & & \\
\hline
\end{tabular}
\end{center}
\caption{Comparison between classical and quantum probability}
\end{table}

%--------------------------------------------------------
\section{Quantum probability and the spectral presheaf}			\label{Sec_QProbabAndSigma}
The analogy between classical and quantum probability can be strengthened if we can find a suitable sample space for the quantum side, in analogy to the Gelfand spectrum $\Sigma_V$ of an abelian von Neumann algebra $V$. Given a nonabelian von Neumann algebra $\cN$, with the self-adjoint operators in $\cN$ representing random variables, such a sample space $\Sigma$ should
\begin{itemize}
	\item generalise the Gelfand spectrum $\Sigma_V$ to the nonabelian von Neumann algebra $\cN$,
	\item come equipped with a family of measurable subsets, analogous to the clopen subsets $\Cl(\Sigma_V)$ of $\Sigma_V$,
	\item serve as a common domain for the random variables, and hence as a common codomain for the associated spectral measures,
	\item serve as a domain for the states of $\cN$, seen as probability measures.
\end{itemize}
The topos approach to quantum theory (see \cite{DoeIsh08a,DoeIsh08b,DoeIsh08c,DoeIsh08d,Doe09a,Doe09b,DoeIsh11,Doe11a,Doe11b,DoeIsh12} and \cite{HLS09a,HLS09b,HLS11}) provides such a generalised sample space, in the form of the spectral presheaf $\Sig$ of a von Neumann algebra $\cN$. We will introduce this presheaf shortly, but first we consider the base category over which it is defined.

\subsection{The context category and the spectral presheaf of a quantum system}
\begin{definition}
Let $\cN$ be a von Neumann algebra, and let $\VN$ be the set of non-trivial abelian von Neumann subalgebras of $\cN$ that share the unit element with $\cN$. When equipped with inclusion as a partial order, $\VN$ is called the \emph{context category of $\cN$}.
\end{definition}
The physical idea is that each context, i.e., each abelian subalgebra $V\in\VN$, provides a \emph{classical perspective} on the quantum system at hand. The poset $\VN$ keeps track of how classical perspectives relate to each other, that is, how they overlap.

Given a von Neumann algebra $\cN$, one can define a Jordan algebra which has the same elements as $\cN$ and multiplication given by
\begin{equation}
			\forall \hA,\hB\in\cN : \hA\cdot\hB:=\frac{1}{2}(\hA\hB+\hB\hA).
\end{equation}
It was shown in \cite{HarDoe10} that the poset $\VN$ determines the algebra $\cN$ as a Jordan algebra, i.e., up to Jordan isomorphisms. 

The probabilistic aspects of quantum theory only depend on the Jordan structure on the set of self-adjoint operators representing physical quantities, as was already shown by Jordan, von Neumann and Wigner \cite{JNW34}. Hence, it is plausible that the probabilistic aspects of quantum theory can be formulated using structures over the poset $\VN$, which determines the Jordan structure of the quantum system described by $\cN$. In the following, we will demonstrate that this is indeed the case.

\begin{definition}
Let $\cN$ be a von Neumann algebra, and let $\VN$ be its context category. The \emph{spectral presheaf $\Sig$ of $\cN$} is the presheaf over $\VN$ given as follows:
\begin{itemize}
	\item [(a)] on objects: for all $V\in\VN$, let $\Sig_V$ be the Gelfand spectrum of $V$, that is, the set of pure states of $V$ equipped with the relative weak$^*$-topology,
	\item [(b)] on arrows: for all inclusions $i_{V'V}: V'\hookrightarrow V$, let $\Sig(i_{V'V})$ be function
	\begin{align}
				\Sig(i_{V'V}):\Sig_V &\lra \Sig_{V'}\\			\nonumber
				\ld &\lmt \ld|_{V'}.
	\end{align}
\end{itemize}
\end{definition}
It is well-known that the restriction map $\Sig(i_{V'V})$ is surjective and continuous. It is also closed and open.

The spectral presheaf $\Sig$, which `glues' together the Gelfand spectra of all abelian subalgebras of $\cN$ in a canonical way, will serve as the quantum sample space. Clearly, $\Sig$ is a generalisation of the Gelfand spectrum $\Sigma_V$ of an abelian von Neumann algebra $V$. Being a presheaf over $\VN$, $\Sig$ is an object in the topos $\SetVNop$ of presheaves over $\VN$. Each topos can be seen as a universe of (generalised) sets, see e.g. \cite{McLMoe92,Joh02/03}, so $\Sig$ is a generalised set in this sense.

Next, we have to specify suitable measurable sub`sets' of $\Sig$, which in this case are subobjects of $\Sig$. A subobject of a presheaf is a subpresheaf.

\begin{definition}
A subobject $\ps S$ of the spectral presheaf $\Sig$ is called \emph{clopen} if, for all $V\in\VN$, the component $\ps S_V$ is clopen in $\Sig_V$, the Gelfand spectrum of $V$. The set of clopen subobjects of $\Sig$ is denoted as $\Subcl\Sig$.
\end{definition}

A clopen subobject $\ps S\subseteq\Sig$ hence is a collection $(\ps S_V)_{V\in\VN}$ of clopen subsets, one in each Gelfand spectrum $\Sig_V$, $V\in\VN$. The condition of being a subobject means that whenever $V',V\in\VN$ such that $V'\subseteq V$, it holds that
\begin{equation}
			\Sig(i_{V'V})(\ps S_V)\subseteq\ps S_{V'}.
\end{equation}
The set $\Subcl\Sig$ of clopen subobjects is partially ordered in an obvious way:
\begin{equation}
			\ps S_1\leq\ps S_2: \quad\Longleftrightarrow\quad (\forall V\in\VN:\ps S_{1;V}\subseteq \ps S_{2;V}).
\end{equation}
With respect to this partial order, all meets and joins exist, so $\Subcl\Sig$ is a complete lattice. Concretely, given a family $(\ps S_i)_{i\in I}$ of clopen subobjects, we have
\begin{align}
			\forall V\in\VN: &(\bjoin_{i\in I}\ps S_i)_V=\on{cl}(\bigcup_{i\in I}\ps S_{i;V}),\\
			&(\bmeet_{i\in I}\ps S_i)_V=\on{int}(\bigcap_{i\in I}\ps S_{i;V}).
\end{align}
Taking the closure of the set-theoretic union respectively the interior of the intersection is necessary in order to obtain clopen subsets at each stage. Note that for each $V$, $\Sig_V$ is extremely disconnected, so the closure of an open set is open, and the interior of a closed set is closed.

The lattice $\Subcl\Sig$ does \emph{not} inherit the orthocomplementation from its components, the complete Boolean algebras $\Cl(\Sig_V)$, $V\in\VN$. Taking the complement in each component of a clopen subobject $\ps S\in\Subcl\Sig$ does not give a subobject (unless $\ps S$ is the empty subobject or $\Sig$ itself). Instead, $\Subcl\Sig$ is a complete Heyting algebra (see Thm. 2.5 in \cite{DoeIsh08b}) with a pseudo-complement
\begin{align}
				\neg:\Subcl\Sig &\lra \Subcl\Sig\\			\nonumber
				\ps S &\lmt \bjoin\{\ps T\in\Subcl\Sig \mid \ps S\wedge\ps T=\ps 0\},
\end{align}
where $\ps 0$ is the empty subobject. In general, $\neg\ps S\vee\ps S\leq\Sig$. In terms of the Heyting implication (which we will not discuss here, but see \cite{DoeIsh08b,Doe12}), the negation is given by $\neg\ps S=(\ps S\Rightarrow\ps 0)$ as usual. 

Interestingly, $\Subcl\Sig$ has even more structure. It is also a complete co-Heyting algebra with a second kind of pseudo-complement, given by
\begin{align}
			\sim:\Subcl\Sig &\lra \Subcl\Sig\\			\nonumber
			\ps S &\lmt \bmeet\{\ps T\in\Subcl\Sig \mid \ps S\vee\ps T=\Sig\}.
\end{align}
In general, $\sim\ps S\wedge\ps S\geq\ps 0$. 

This makes $\Subcl\Sig$ into a complete bi-Heyting algebra. We remark that every Boolean algebra is a bi-Heyting algebra in which both kinds of negation coincide. For more details, see \cite{Doe12}.

The complete bi-Heyting algebra $\Subcl\Sig$ of clopen subobjects of the spectral presheaf $\Sig$ of a noncommutative von Neumann algebra is the quantum analogue of the complete Boolean algebra $\Cl(\Sigma_{\VlA})$ of clopen subsets of the Gelfand spectrum $\Sigma_{\VlA}$ of a commutative von Neumann algebra $\VlA$. 

Importantly, the spectral presheaf $\Sig$ has no global elements if $\cN$ has no summand of type $I_2$, that is, one cannot choose one element $\ld_V$ from each component $\Sig_V$ such that, whenever $V'\subset V$, it would hold that
\begin{equation}
			\Sig(i_{V'V})(\ld_V)=\ld_{V'}.
\end{equation}
In fact, such a choice of elements $(\ld_V)_{V\in\VN}$ would give a \emph{valuation function}, that is, a function $v:\cN_{\sa}\ra\bbR$ such that for all $\hA\in\cN_{\sa}$, it would hold that $v(\hA)\in\sp\hA$, and for all bounded Borel functions $f:\bbR\ra\bbR$, one would have $v(f(\hA))=f(v(\hA))$.

The Kochen-Specker theorem shows that such valuation functions do not exist if $\cN$ has no summand of type $I_2$ (see \cite{KocSpe67}; for the generalisation to von Neumann algebras, see \cite{Doe05}). Isham and Butterfield observed in \cite{IshBut98,IshBut00} that the Kochen-Specker theorem is equivalent to the fact that the spectral presheaf $\Sig$ has no global elements.

Since global elements of a presheaf are the analogues of elements of a set (or points of a space), the quantum sample space $\Sig$ is a space without points in this sense. The fact that $\Sig$ has no points is ultimately the reason why we can reformulate quantum probability based on the set-like object $\Sig$, without falling prey to the Kochen-Specker theorem.\footnote{Related to this, pure states in the topos-based formulation of quantum theory do not correspond to points of the sample space (as they do in classical theory), quite simply because the sample space $\Sig$ has no points. We do not discuss the representation of pure states here; see \cite{DoeIsh11}. The representation of general states is given in subsection \ref{Subsec_StatesAsProbabMeasures} below.}

%Just as $\Sigma_{\VlA}$ served as the sample space in the commutative case, $\Sig$ will serve as the sample space in the noncommutative case, jointly for all quantum random variables.

\subsection{Random variables and their inverse images}			\label{SubSec_QRandomVars}
In the perspective of the topos approach to quantum theory, a random variable is a generalised function from the sample space $\Sig$ (which in fact is a presheaf) to a presheaf of real numbers. There is a representation of self-adjoint operators $\hA\in\cN_{\sa}$ by arrows
\begin{equation}			\label{Def_bdas(A)}
			\bdas{\hA}:\Sig \lra \Rlr
\end{equation}
in the topos $\SetVNop$. This representation was discussed in detail elsewhere \cite{DoeIsh08c,Doe11b}, so we will not give details here. In a nutshell, $\Sig$ physically plays the role of a (generalised) state space of the quantum system, while $\Rlr$ is a space of values of physical quantities, generalising the real numbers in the sense that not only sharp, definite real values exist, but also `unsharp' values in the form of real intervals. The definition of the arrow $\bdas{\hA}$ is based on the maps called \emph{inner} and \emph{outer daseinisation of self-adjoint operators}, which are approximations with respect to the spectral order on self-adjoint operators. For details, see \cite{Doe11b} and Part I \cite{DoeDew12a}.

The arrow $\bdas{\hA}$ is the analogue of a function $f_A$ from the state space to the real numbers. (Such a function $f_A$ represents a physical quantity $A$ in the classical case). Mathematically, $\bdas{\hA}$ generalises the Gelfand transform $\ol A:\Sigma_V\ra\bbR$ of a self-adjoint operator in an abelian von Neumann algebra $V$. In a probabilistic view of quantum theory, we can interpret the arrow $\bdas{\hA}$ as a random variable (associated with the physical quantity $A$).

But note that here, both the sample space $\Sig$ and the space of values $\Rlr$ are presheaves, i.e., objects in the topos of presheaves, and the arrow $\bdas{\hA}:\Sig\ra\Rlr$ is a natural transformation, that is, an arrow in the topos. In the following, we will consider a more modest approach, using only the generalised sample space $\Sig$, while the space of values will be the traditional real numbers $\bbR$ or extended real numbers $\eR$.

\noindent\textbf{From projection-valued measures to clopen subobject-valued measures.} We will not be concerned with the representation of random variables per se.\footnote{They can be defined as suitable partial functions from $\Sig$, or rather, the disjoint union of its fibres, to $\eR$.} As in the rest of the paper, the `inverse image map', from Borel subsets of $\eR$ to $\Subcl\Sig$, the measurable subobjects of $\Sig$, will be more important. That is, the inverse image maps measurable subsets of outcomes to measurable subsets of the sample space. A probability measure then maps measurable subsets of the sample space to probabilities.

We will relate a (topos-external) Borel set $\De$ of outcomes to a clopen subobject of the (topos-internal) sample space $\Sig$ via a `translation map' (which in fact is the map called the outer daseinisation of projections, see Def. \ref{Def_OuterDasOfProjs}). We will show that this construction, together with the representation of quantum states as probability measures on $\Sig$, reproduces the Born rule and hence the usual quantum mechanical predictions of expectation values. %For some more conceptual discussion, see the end of this section.

Given a self-adjoint operator $\hA$, we obtain a projection $\hP=\eA(\De)$ by the spectral theorem, where $\eA:\cB(\eR)\ra\PN$ is the (extended) spectral measure of $\hA$, see \eq{Def_SpecMeasOnExtendedReals}.\footnote{In previous papers by AD and C. Isham, the notation $\hE[\Ain\De]$ had been used for $\eA(\De)$.} In standard quantum theory, this projection represents the proposition ``$\Ain\De$'', that is, ``if a measurement of the physical quantity $A$ is performed, the measurement outcome will be found to lie in the Borel set $\De$''.

\begin{remark}
In the topos approach to quantum theory, the proposition ``$\Ain\De$'' is often interpreted in a more `realist' or `ontological' manner to mean ``the physical quantity $A$ has a value in the set $\De$''. The truth value of such a proposition in a given state $\rho$ of the quantum system in general is neither (totally) true nor (totally) false. The key idea is to use the richer, intuitionistic internal logic of the topos $\SetVNop$ to assign a truth value -- and not just a probability -- to the proposition, independent of measurements and observers. This is described in detail for pure states in \cite{DoeIsh11} and for mixed states in \cite{DoeIsh12}. In the present article, we are concerned with aspects of a probabilistic interpretation of quantum theory, so we focus on the more `instrumentalist' or `operational' interpretation of ``$\Ain\De$'' so that it means ``if a measurement of $A$ is performed, the outcome will lie in $\De$''.
\end{remark}

Next, we have to relate projections to clopen subobjects of the sample space $\Sig$. Let $\hP\in\PN$ be a projection, and let $V\in\VN$ be a context. We approximate $\hP$ in $V$ (from above) by
\begin{equation}			\label{Def_OuterDasOfPToV}
			\doto{\hP}{V}:=\bmeet\{\hQ\in\PV \mid \hP\leq\hQ\}.
\end{equation}
The projection $\doto{\hP}{V}$ is called the outer daseinisation of $\hP$ to $V$. Using the isomorphism $\alpha_V:\PV\ra\Cl(\Sig_V)$ (see \eq{Def_alphaV}), the projection $\doto{\hP}{V}$ in $V$ corresponds to the clopen subset 
\begin{equation}
			S_{\doto{\hP}{V}}:=\alpha_V(\doto{\hP}{V})
\end{equation}
of $\Sig_V$. Now, taking \emph{all} contexts $V\in\VN$ into account at the same time, we define:

\begin{definition}			\label{Def_OuterDasOfProjs}
The map
\begin{align}
			\ps\de:\PN &\lra \Subcl\Sig\\			\nonumber
			\hP &\lmt \pde{\hP}=(S_{\doto{\hP}{V}})_{V\in\VN}
\end{align}
is called \emph{outer daseinisation of projections}.
\end{definition}

Hence, the map $\ps\de$ maps a projection $\hP$ in a von Neumann algebra $\cN$ to a clopen subobject $\pde{\hP}$ of the spectral presheaf $\Sig$ of the algebra. Locally, at each abelian subalgebra $V\in\VN$, the map $\ps\de$ is given by the outer daseinisation $\doto{\hP}{V}$ of $\hP$ to $V$ (see \eq{Def_OuterDasOfPToV} above), which justifies calling the map $\ps\de$ by the same name.

The map $\ps\de$ was introduced in \cite{DoeIsh08b} and discussed further in \cite{Doe11a,Doe11b,Doe12}. Its main properties are:
\begin{itemize}
	\item [(a)] $\ps\de$ is monotone, $\pde{\hat 0}=\ps\emptyset$ and $\pde{\hat 1}=\Sig$,
	\item [(b)] $\ps\de$ is injective, but not surjective,
	\item [(c)] for all families $(\hP_i)_{i\in I}$ of projections in $\cN$, it holds that $\pde{\bjoin_i \hP_i}=\bjoin_i \pde{\hP_i}$, that is, $\ps\de$ preserves all joins,
	\item [(d)] for all $\hP,\hQ\in\PN$, $\pde{\hP\meet\hQ}\leq\pde{\hP}\meet\pde{\hQ}$.
\end{itemize}

If the projection $\hP$ is a spectral projection of some self-adjoint operator $\hA$, that is, if $\hP=\eA(\De)$, we write
\begin{equation}
			\ps S(\hA;\De):=\pde{\eA(\De)}.
\end{equation}
Note that if we consider a context $V$ such that $\hA\in V$, then all the spectral projections $\eA(\De)$ are contained in $V$, so $\doto{\eA(\De)}{V}=\eA(\De)$. Hence, the component of the clopen subobject $\ps S(\hA;\De)$ at $V$ is given by
\begin{equation}
			\ps S(\hA;\De)_V=S_{\doto{\eA(\De)}{V}}=\alpha_V(\doto{\eA(\De)}{V})=\alpha_V(\eA(\De))
\end{equation}
which is the clopen subset of $\Sig_V$ corresponding to the usual spectral projection under the isomorphism $\alpha_V:\PV\ra\Cl(\Sig_V)$. 

If we speak of the measurement of a physical quantity $A$ in standard quantum theory, we of course (implicitly) pick a context $V$ that contains $\hA$, the self-adjoint operator representing $A$. In such contexts, the clopen subobject $\ps S(\hA;\De)$ corresponds to the usual spectral projection $\eA(\De)$. Yet, different from standard quantum theory, we also define approximations to a spectral projection of $\hA$ in those contexts $V$ that do not contain $\hA$, in the form of the coarse-grained (i.e., larger) projections $\doto{\eA(\De)}{V}$ or their corresponding clopen subsets $S_{\doto{\eA(\De)}{V}}$. The clopen subobject $\ps S(\eA(\De))$ collates all these approximations into one object, which can be interpreted as a measurable subobject of the quantum sample space $\Sig$.

Summing up, we have defined a map
\begin{equation}
			\breve A^{-1}: \cB(\eR) \lra \Subcl\Sig
\end{equation}
from Borel subsets of the (extended) real line to clopen subobjects of the joint sample space, given by
\begin{equation}
			\breve A^{-1}(\De)=\pde{\eA(\De)}=\ps S(\hA;\De).
\end{equation}

In the new topos-based perspective we are developing here, $\breve A^{-1}$ is the inverse image of the random variable (physical quantity) $A$. We have replaced the projection-valued spectral measure $\eA:\cB(\eR)\ra\PN$ by a clopen subobject-valued measure $\breve A^{-1}:\cB(\eR)\ra\Subcl\Sig$. 

\begin{remark}
A full topos-internal treatment of quantum random variables and their inverse images would proceed by using the arrow $\bdas{\hA}:\Sig \lra \Rlr$ (see \eqref{Def_bdas(A)}) and its `inverse image', which involves pullbacks of suitably defined `clopen subobjects' of $\Rlr$.
\end{remark}

\subsection{$\Subcl\Sig$-valued CDFs and quantile functions}
Given a random variable, described by a self-adjoint operator $\hA$ or the corresponding inverse image map $\breve A^{-1}:\cB(\eR)\ra\Subcl\Sig$, we can now define a clopen subobject-valued version of the cumulative distribution function:
\begin{align}			\label{Def_EbrA}
			\EbrA:\eR &\lra \Subcl\Sig\\			\nonumber
			r &\lmt \breve A^{-1}([-\infty,r])=\pde{\hEA_r}=\ps S(\hA;[-\infty,r]).
\end{align}
Note that we can include $-\infty$ in the sets $[-\infty,r]$ since we define our map over the extended reals. Clearly, $\EbrA$ is monotone and $\EbrA(-\infty)=\ps\emptyset$, the empty subobject of $\Sig$ (since $\hEA_{-\infty}=\hat 0$), and $\EbrA(\infty)=\Sig$ (since $\hEA_{\infty}=\hat 1$).

\begin{proposition}			\label{Prop_EbrARightCont}
Let $\cN$ be a von Neumann algebra, and let $\hA\in\cN_{\sa}$ a self-adjoint operator. The map $\EbrA:\eR\ra\Subcl\Sig$ is right-continuous.
\end{proposition}

\begin{proof}
We already saw that the bottom element $-\infty=\bmeet_{r\in\eR} r$ is mapped to the bottom element $\ps\emptyset=\bmeet_{\ps S\in\Subcl\Sig} \ps S$, and that the top element $\infty=\bmeet\emptyset$ is mapped to the top element $\Sig=\bmeet\emptyset$.

Let $r\in\bbR$. For each $V\in\VN$, we have
\begin{align}
			(\bmeet_{s>r} \pde{\hEA_s})_V &= \bmeet_{s>r} \doto{\hEA_s}{V}\\
			&= \bmeet_{s>r} \bmeet\{\hQ\in\PV \mid \hEA_s\leq\hQ\}\\
			&= \bmeet\{\hQ\in\PV \mid \bmeet_{s>r}\hEA_s\leq\hQ\}\\
			&= \bmeet\{\hQ\in\PV \mid \hEA_r\leq\hQ\}\\
			&= \doto{\hEA_r}{V} = \pde{\hEA_r}_V,
\end{align}
so
\begin{equation}
			\bmeet_{s>r} \EbrA(s) = \bmeet_{s>r} \pde{\hEA_s} = \pde{\bmeet_{s>r} \hEA_s} =\pde{\hEA_r} = \EbrA(r).
\end{equation}
\end{proof}
We remark that the daseinisation map $\ps\de:\PN\ra\Subcl\Sig$ does \emph{not} preserve meets in general, but it preserves the particular meets needed here.

We have shown that $\EbrA:\eR\ra\Subcl\Sig$ is the cumulative distribution function (CDF) of the inverse image $\breve A^{-1}$ of the random variable $A$. In the notation of subsection \ref{Subsec_LValuedMeasures}, $\breve A^{-1}:\cB(\eR)\ra\Subcl\Sig$ is an $L$-valued measure, and $\EbrA$ is its $L$-CDF, where $L$ is the complete meet-semilattice $\Subcl\Sig$.

It follows easily from Prop. \ref{Prop_EbrARightCont} that $\EbrA$ is a meet-preserving map between complete lattices, so it has a left adjoint $\obrA$, given by
\begin{align}
			\obrA:\Subcl\Sig &\lra \eR\\			\nonumber
			\ps S &\lmt \inf\{r\in\eR \mid \ps S\leq\EbrA(r)\}.
\end{align}
This function is the $\Subcl\Sig$-quantile function of the $\Subcl\Sig$-CDF $\EbrA$. If we consider a clopen subobject $\ps S\in\Subcl\Sig$ of the form $\ps S=\pde{\hP}$, then
\begin{align}
			\obrA(\pde{\hP}) &= \inf\{r\in\eR \mid \pde{\hP}\leq\EbrA(r)\}\\
			&= \inf\{r\in\eR \mid \pde{\hP}\leq\pde{\hEA_r}\}\\
			&= \inf\{r\in\eR \mid \hP\leq\hEA_r\}\\
			&= \oA(\hP),
\end{align}
so the function $\obrA:\Subcl\Sig\ra\eR$ generalises the $q$-observable function $\oA:\PN\ra\eR$, which is the $\PN$-quantile function of the $\PN$-valued measure $\eA$.

Just as $\breve A^{-1}$, the inverse image of the random variable, is constructed in two steps, the CDF $\EbrA$ is as well: first, we have the map $\EA:\eR\ra\PN$ (the extended spectral family of $\hA$), then, the outer daseinisation of projections $\ps\de:\PN\ra\Sig$, so
\begin{equation}
			\EbrA=\ps\de\circ\EA:\eR \lra \Subcl\Sig.
\end{equation}

\subsection{States as probability measures on the spectral presheaf}			\label{Subsec_StatesAsProbabMeasures}
We briefly discuss the representation of quantum states as probability measures on the spectral presheaf, introduced in \cite{Doe09a} and further developed in \cite{DoeIsh12}. For some related work, see \cite{HLS11}.

Let $\rho$ be a state of the von Neumann algebra $\cN$ of physical quantities. Let $\cA(\VN,[0,1])$ denote the set of antitone (i.e., order-reversing) functions from the context category $\VN$ to the unit interval.

\begin{definition}
The \emph{probability measure on the spectral presheaf $\Sig$ associated with $\rho$} is the map
\begin{align}
			\mu_\rho:\Subcl\Sig &\lra \cA(\VN,[0,1])\\			\nonumber
			\ps S=(\ps S_V)_{V\in\VN} &\lmt (\rho(\alpha_V^{-1}(\ps S_V)))_{V\in\VN}.
\end{align}
\end{definition}
Consider two contexts such that $V'\subset V$. The fact that $\ps S$ is a (clopen) subobject is equivalent to the fact that $\alpha_{V'}^{-1}(\ps S_V)\geq\alpha_V^{-1}(\ps S_V)$. Hence,
\begin{equation}
			(\mu_\rho(\ps S))_{V'}=\rho(\alpha_{V'}^{-1}(\ps S_V))\geq\rho(\alpha_V^{-1}(\ps S_V))=(\mu_\rho(\ps S))_V,
\end{equation}
so $\mu_\rho$ indeed is an antitone function. Note that we assign a probability $(\mu_\rho(\ps S))_V\in [0,1]$ to every context $V\in\VN$ such that smaller contexts are assigned larger probabilities (which is a consequence of coarse-graining).

It is straightforward to show that for all states $\rho$ of $\cN$,
\begin{itemize}
	\item [(1.)] $\mu_\rho(\ps\emptyset)=0_{\VN}$, where $0_{\VN}$ is the function that is constantly $0$ on $\VN$, and analogously, $\mu_\rho(\Sig)=1_{\VN}$,
	\item [(2.)] for all $\ps S_1,\ps S_2\in\Subcl\Sig$,
	\begin{equation}
				\mu_\rho(\ps S_1)+\mu_\rho(\ps S_2)=\mu_\rho(\ps S_1)\join\mu_\rho(\ps S_2)+\mu_\rho(\ps S_1)\meet\mu_\rho(\ps S_2),
	\end{equation}
	where addition (and meets and joins) are taken contextwise.
\end{itemize}
Any map
\begin{align}
			\mu:\Subcl\Sig &\lra \cA(\VN,[0,1])\\			\nonumber
			\ps S=(\ps S_V)_{V\in\VN} &\lmt (\mu_V(\ps S_V))_{V\in\VN}
\end{align}
such that (1.) and (2.) hold is called a \emph{probability measure on $\Sig$}. Note that $\mu$ has components $\mu_V$, $V\in\VN$, where $\mu_V:(\Subcl\Sig)_V\simeq\Cl(\Sig_V)\ra [0,1]$. If $\mu_1,\mu_2$ are probability measures on $\Sig$ and $c\in [0,1]$, then we define the convex combination $c\mu_1+(1-c)\mu_2$ by
\begin{equation}
			\forall V\in\VN: (c\mu_1+(1-c)\mu_2)_V = c\mu_{1;V}+(1-c)\mu_{2;V}.
\end{equation}

It was shown in \cite{Doe09a} that if $\cN$ has no type $I_2$-summand, every probability measure $\mu$ on $\Sig$ determines a unique quantum state $\rho_\mu$ of $\cN$ such that $\mu_{\rho_\mu}=\mu$ and $\rho_{\mu_\rho}=\rho$.

\begin{theorem}
Let $\cN$ be a von Neumann without a summand of type $I_2$. There is a bijection between the convex set $\mc S(\cN)$ of states of $\cN$ and the convex set $\mc M(\Sig)$ of probability measures on its spectral presheaf.
\end{theorem}
A probability measure $\mu\in\mc M(\Sig)$ corresponds to a normal state $\rho_\mu$ if and only if it preserves joins of increasing nets of clopen subobjects, see \cite{DoeIsh12}.

The codomain of a probability measure $\mu_\rho$ is the set $\cA(\VN,[0,1])$ of antitone functions from $\VN$ to $[0,1]$. This set is in fact isomorphic to the set of global sections of a certain presheaf $\ps{\bbR^\succeq}$ (see \cite{Doe09a}). Seen topos-internally, this presheaf is the presheaf of upper reals in $\SetVNop$. Hence, a probability measure $\mu_\rho$ takes its values in the global sections of a suitable topos-internal space of `values'. Intuitively speaking, there is one probability for each context. From the usual topos-external perspective, we expect to obtain just a single probability. In the next subsection, we show how to do this explicitly.

\subsection{State-proposition pairing and the Born rule} 
Let ``$\Ain\De$'' be a proposition, which we interpret as ``upon measurement, the outcome of the random variable $A$ will be found to lie in $\De$''. In subsection \ref{SubSec_QRandomVars}, we had associated a clopen subobject of $\Sig$ with ``$\Ain\De$'', namely $\breve A^{-1}(\De)=\ps S(\hA;\De)=\pde{\eA(\De)}$. 

Given a quantum state $\rho$, we defined the associated probability measure $\mu_\rho:\Subcl\Sig\ra\mc A(\VN,[0,1])$. Hence, we can define a pairing between propositions and states by
\begin{equation}			\label{Def_StatePropPairing}
			\text{(``}\Ain\De\text{''},\rho) \lmt \mu_\rho(\breve A^{-1}(\De)).
\end{equation}
This is exactly analogous to the classical case: there, a physical quantity $A$ is represented by a measurable function $f_A:\Sigma\ra\bbR$ from the state space $\Sigma$ of the system to the real numbers, and a proposition ``$\Ain\De$'' is represented by the measurable subset $f_A^{-1}(\De)\subseteq\Sigma$ of the state space. Moreover, a classical state is a probability measure $\mu_c:\cB(\Omega)\ra [0,1]$, where $\cB(\Omega)$ denotes the measurable subsets of $\Omega$. The state-proposition pairing in the classical case is given by
\begin{equation}
			\text{(``}\Ain\De\text{''},\rho) \lmt \mu_c(f_A^{-1}(\De)).
\end{equation}
Yet, in the quantum case we have a function $\mu_\rho(\breve A^{-1}(\De)):\VN\ra [0,1]$, assigning one probability to each context $V\in\VN$, while classically, we have just one probability. 

Moreover, the usual Born rule of quantum theory gives only one probability: the Born rule says that in the state $\rho$, the probability that upon measurement of $A$ the outcome will lie in $\De$ is given by
\begin{equation}
			\on{Prob}(\text{``}\Ain\De\text{''};\rho)=\rho(\eA(\De)).
\end{equation}
As mentioned before, measuring of $A$ implicitly means choosing a context that contains $A$. Let $V$ be a context that contains $\hA$, the self-adjoint operator representing $A$. As we saw in subsection \ref{SubSec_QRandomVars}, in such a context $V$ the component $(\breve A^{-1}(\De))_V=\ps S(\hA;\De)$ is given by $\eA(\De)$ (or the corresponding clopen subset $S_{\eA(\De)}\subseteq\Sig_V$). Hence, we obtain
\begin{equation}
			(\mu_\rho(\breve A^{-1}(\De)))_V=(\mu_\rho)_V(S_{\eA(\De)})=\rho(\eA(\De))
\end{equation}
if $\hA\in V$. In this way, we reproduce the Born rule, the predictive content of quantum theory.

Moreover, if $\hA\notin V$, then the component $(\breve A^{-1}(\De))_V=\ps S(\hA;\De)$ is given by $\doto{\eA(\De)}{V}$, which in general is a larger projection than $\eA(\De)$. Hence, in such a context we obtain
\begin{equation}
			(\mu_\rho(\breve A^{-1}(\De)))_V=(\mu_\rho)_V(S_{\doto{\eA(\De)}{V}})=\rho(\doto{\eA(\De)}{V}),
\end{equation}
which is larger than or equal to the probability determined by the Born rule. (This is another consequence of coarse-graining.) We have shown:

\begin{proposition}
The minimum of the function $\mu_\rho(\breve A^{-1}(\De)):\VN\ra [0,1]$ resulting from our state-proposition pairing \eq{Def_StatePropPairing} is the usual quantum expectation value $\rho(\eA(\De))$ provided by the Born rule.
\end{proposition}

The quantum-theoretical CDF of a quantum random variable with respect to a state $\rho$ is obtained as the composition
\begin{equation}
			C^{\breve A}:=\min_{V\in\VN}\circ\;\mu_\rho\circ\EbrA,
\end{equation}
where $\mu_\rho:\Subcl\Sig\ra\cA(\VN,[0,1])$ is the probability measure on $\Sig$ corresponding to $\rho$. Note that $C^{\breve A}$ takes values in $[0,1]$, just as a CDF usually does. Concretely,
\begin{align}
			C^{\breve A}:\eR &\lra [0,1]\\			\nonumber
			r &\lmt \min_{V\in\VN} \mu_{\rho}(\EbrA(r)).
\end{align}
By definition, $\EbrA(r)=\breve A^{-1}([-\infty,r])=\pde{\hEA_r}$ (cf. \eq{Def_EbrA}). The minimum is attained at those contexts $V\in\VN$ that contain $\hA$, which confirms that $C^{\breve A}$ indeed is the CDF of the random variable described by $\hA$ (or its Gelfand transform $\ol A$).

Finally, the corresponding quantile function is the left adjoint of $C^{\breve A}$, that is,
\begin{align}
			q^{\breve A}:[0,1] &\lra \eR\\			\nonumber
			s &\lmt \inf\{r\in\eR \mid s\leq C^{\breve A}(r)\}.
\end{align}

We summarise our results in table 2, which should be compared to table 1 above. In table 2, we use the Gelfand spectrum $\Sigma_V$ of an abelian von Neumann algebra $V$ as the sample space on the classical side. Since the Gelfand transform $\ol A:\Sigma_V\ra\sp\hA\subset\eR$ of a self-adjoint operator $\hA\in V$ is the representative of a random variable in the classical case, the inverse image of the random variable is a map $\ol A^{-1}:\cB(\eR)\ra\Cl(\Sigma_V)$. On the quantum side, the spectral presheaf $\Sig$, which is a reasonably straightforward generalisation of the Gelfand spectrum to nonabelian von Neumann algebras, plays the role of the sample space.

The relevant complete meet-semilattices are $L=\Cl(\Sigma_V)$ (which is a complete Boolean algebra isomorphic to $\PV$, see \eq{Def_alphaV}) in the classical case and $L=\Subcl\Sig$ (which is a complete bi-Heyting algebra) in the quantum case.

\begin{table}[htdp]
\label{Table_CAndQProbab2}
\begin{center}
\begin{tabular}{| c | c | c |} \hline
 & & \\
 & Classical & Quantum \\
 & & \\
\hline & & \\
Sample space & $\Sigma_V$ & $\Sig$\\
 & & \\
\hline & & \\
Inv. im. of random var. & $\ol A^{-1}:\cB(\eR)\ra\Cl(\Sigma_V)$ & $\breve A^{-1}:\cB(\eR) \ra \Subcl\Sig$\\
 & $\De \mt \ol A^{-1}(\De)$ & $\De \mt \pde{\eA(\De)}=\ps S(\hA;\De)$\\
 & & \\
\hline & & \\
$L$-CDF & $\tCA:\eR \ra \Cl(\Sigma_V)$ & $\EbrA:\eR\ra\Subcl\Sig$\\
 & $r \mt \tCA(r)=\ol A^{-1}([-\infty,r])$ & $r \mt \EbrA_r=\breve A^{-1}([-\infty,r])$\\
 & & \\
\hline & & \\
$L$-quantile function & $\tqA:\Cl(\Sigma_V)\ra\eR$ & $\tilde q^{\breve A}=\obrA:\Subcl\Sig\ra\eR$\\
 & $S \mt \inf\{r\in\eR \mid S\leq\tCA(S)\}$ & $\ps S \mt \inf\{r\in\eR \mid \ps S\leq\EbrA(r)\}$\\
 & & \\
\hline & & \\
State & $\mu:\Cl(\Sigma_V)\ra [0,1]$ & $\mu_\rho:\Subcl\Sig\ra \cA(\VN,[0,1])$\\
 & & \\
\hline & & 
\\CDF & $C^A=\mu\circ\tCA:\eR\ra [0,1]$ & $C^{\breve A}:\eR\ra [0,1]$\\
 & $r \mt \mu(A^{-1}([-\infty,r]))$ & $r \mt \min_{V\in\VN}\mu_\rho(\breve A^{-1}([-\infty,r]))$ \\
 & & \\
\hline & & \\
Quantile function & $q^A:[0,1]\ra\eR$ & $q^{\breve A}:[0,1]\ra\eR$\\
 & $s \mt \inf\{r\in\eR \mid s\leq C^A(r)\}$ & $s \mt \inf\{r\in\eR \mid s\leq C^{\breve A}(r)\}$ \\
 & & \\
\hline
\end{tabular}
\end{center}
\caption{Classical probability and quantum probability in the topos-based perspective}
\end{table}

%--------------------------------------------------------
\section{Summary and outlook}			\label{Sec_Summary}
\textbf{$L$-valued measures, CDFs, and quantile functions.} We observed in section \ref{Sec_qObsFctsAndProbab} that even in the classical case of a random variable $A$ and a probability measure $\mu$ on a measurable space $\Omega$, there is a Galois connection between the cumulative distribution function $C^A$ and the quantile function $q^A$. This becomes obvious from the adjoint functor theorem for posets when the domain of $C^A$ (and hence the codomain of $q^A$) is extended canonically to the extended reals $\eR$.

This observation was then generalised to measures taking values not in $[0,1]$, but in a complete meet-semilattice $L$. This includes 
\begin{itemize}
	\item the case $L=\cB(\Omega)$ of (equivalence classes of) measurable subsets (modulo null subsets) of a classical sample space $\Omega$, where the $L$-valued measure is the inverse image map $A^{-1}$ of a random variable $A$, 
	\item the usual quantum-mechanical case of $L=\PN$, that is, of projection-valued measures, where the $L$-valued measure of a random variable $A$, represented by a self-adjoint operator $\hA$, is the spectral measure $\eA$,
	\item the presheaf- and topos-based case of $L=\Subcl\Sig$ of clopen subobjects of the spectral presheaf $\Sig$, which plays the role of a joint sample space for all quantum random variables, as shown in section \ref{Sec_QProbabAndSigma} (and summarised below).
\end{itemize}
It is not important if $L$ is a distributive lattice (such as $\cB(\Omega)$ in the classical case), or a non-distributive one (such as $\PN$ in the quantum case), it is just required that $L$ be a complete meet-semilattice.\footnote{ We recall that each complete meet-semilattice also is a complete join-semilattice.} $L$-CDFs can be defined for $L$-valued measures in the obvious way, see \eq{Def_L-CDF}. Since $L$-CDFs are meet-preserving maps, they have left adjoints, which we call $L$-quantile functions. An ordinary CDF can be decomposed into an $L$-CDF, followed by the measure, see \eq{Eq_CDFDecomposed}. Analogously, each ordinary quantile function can be decomposed into the adjoint of a measure, followed by the $L$-quantile function, see \eq{Eq_QuantFctDecomposed}. This presupposes that the measure preserves meets, which is a regularity condition that for example holds for normal quantum states.

We then showed that $q$-observable functions, introduced in Part I \cite{DoeDew12a}, are $\PN$-quantile functions for $\PN$-valued measures -- that is, spectral measures -- defined by quantum observables seen as quantum random variables. This gives a clear interpretation of $q$-observable functions in terms of quantum probability theory. In this perspective, the (extended) spectral family $\EA$ plays the role of the $\PN$-CDF.

\textbf{Quantum probability in a topos-based perspective.} We have shown that the spectral presheaf $\Sig$ of the von Neumann algebra $\cN$ of observables of a quantum system can serve as a joint sample space for all quantum observables of the system. The Kochen-Specker theorem is a no-go theorem that shows that no measurable space in the usual, \emph{set}-based sense can play the role of a joint sample space for non-compatible quantum observables (represented by non-commuting self-adjoint operators). In the topos approach to quantum theory, the Kochen-Specker theorem is circumvented by the fact that $\Sig$ is not simply a set but a presheaf. Moreover, the spectral presheaf is a natural generalisation of the Gelfand spectrum $\Sigma_V$ of an abelian von Neumann algebra $V$. As is well-known, $\Sigma_V$ can serve as a sample space in the classical case, that is, for co-measurable physical quantities/random variables, which can be represented by commuting self-adjoint operators (or their Gelfand transforms).

In the classical case, the clopen subsets of $\Sigma_V$ form a complete Boolean algebra $\Cl(\Sigma_V)$ and serve as measurable subsets and hence as the codomain of inverse images of random variables, $A^{-1}:\cB(\im A)\ra\Cl(\Sigma_V)$. In the presheaf-based formulation of quantum probability, the clopen subobjects of $\Sig$ form a complete bi-Heyting algebra $\Subcl\Sig$ that generalises the complete Boolean algebra $\Cl(\Sigma_V)$ of clopen subsets of the Gelfand spectrum. 

Conceptually, a bi-Heyting algebra can be seen as a mild generalisation of a Boolean algebra: a bi-Heyting algebra is a distributive lattice (in contrast to the usual non-distributive lattice $\PN$ of projections used in standard quantum probability). The only difference between a Boolean algebra and a bi-Heyting algebra lies in the way complements are defined: a bi-Heyting algebra has two kinds of negations or pseudo-complements, one coming from the Heyting structure, the other from the co-Heyting structure. In a Boolean algebra, which also is a bi-Heyting algebra, the two negations coincide. (For a logical interpretation of the bi-Heyting structure of $\Subcl\Sig$, see \cite{Doe12}.)

Inverse images of quantum random variables are represented by mappings $\breve A^{-1}:\cB(\eR)\ra\Subcl\Sig$ from standard, topos-external Borel subsets of outcomes to clopen subobjects of the topos-internal sample space $\Sig$. The clopen subobject-valued map $\breve A^{-1}$ generalises the usual projection-valued spectral measure $\eA:\cB(\eR)\ra\PN$. The associated cumulative distribution function $\EbrA:\eR\ra\Subcl\Sig$ and its left adjoint $\obrA:\Subcl\Sig\ra\eR$ are the analogues of the spectral family $\EA:\eR\ra\PN$ and the quantum quantile function $\oA:\PN\ra\eR$. States of the quantum system are given by probability measures $\mu_\rho:\Subcl\Sig\ra\mc A(\VN,[0,1])$.

Crucially, the Born rule is captured by the new formalism, and quantum mechanical probabilities are calculated in a way completely analogous to the classical case: propositions ``$\Ain\De$'' are mapped to measurable subobjects $\ps S(\hA;\De)$ of the quantum sample space $\Sig$, and each quantum state, given by a probability measure $\mu_\rho$ on $\Sig$, assigns an element	$\mu_\rho(\ps S(\hA;\De))$ of $\mc A(\VN,[0,1])$ to each $\ps S\in\Subcl\Sig$. The minimum of the function $\mu_\rho(\ps S(\hA;\De))$ is the usual quantum mechanical probability that upon measurement of the random variable (physical quantity) $A$, the outcome will lie in $\De$:
\begin{equation}
			\on{Prob}(\text{``}\Ain\De\text{''};\rho)=\min_{V\in\VN}(\mu_\rho(\ps S(\hA;\De))).
\end{equation}
The minimum is attained at all contexts $V$ that contain the self-adjoint operator $\hA$.

Summing up, we have derived a formulation of quantum probability that is \emph{structurally} very similar to classical probability, more so than the usual, Hilbert space-based formulation.

\textbf{Outlook.} As is well-known, the probabilistic aspects of quantum theory can naturally be formulated and generalised using positive operator-valued measures (POVMs, also called semispectral measures) \cite{Hol82,NieChu00}. In future work, it will be interesting to consider generalisations of our constructions to POVMs, and to relate our constructions to more advanced aspects of noncommutative probability and their applications in the Fock space formalism of quantum theory \cite{Par92}.

%--------------------------------------------------------

\textbf{Acknowledgements:} We thank Prakash Panangaden and Chris Isham for discussions. We also thank Masanao Ozawa, Izumi Ojima, Mikl\'os R\'edei, Pekka Lahti, Massoud Amini and John Harding for their kind interest and for further suggestions. We hope to be able to develop some of these suggestions in future work.

%--------------------------------------------------------

\end{document}